\documentclass[11pt]{article}
\usepackage{amsfonts}
\usepackage{graphics,amssymb,amsthm,amsmath,epsf}
\usepackage{graphicx}
\usepackage{subfigure}

\usepackage[round]{natbib}

\newtheorem{theorem}{Theorem}[section]
\newtheorem{prop}{Proposition}

\numberwithin{equation}{section}
\newtheorem{remark}{Remark}[section]

\allowdisplaybreaks

\begin{document}
\title{On Bivariate Generalized Linear Failure Rate-Power \\ Series Class of Distributions}
\author{Rasool Roozegar\thanks{Corresponding: rroozegar@yazd.ac.ir } , Ali Akbar Jafari  \ \\
{\small Department of Statistics, Yazd University, Yazd,  Iran}\\}
\date{}
\maketitle

\begin{abstract}
Recently it has been observed that the bivariate generalized linear failure rate distribution can be used quite effectively to analyze lifetime data in two dimensions. This paper introduces a more general class of bivariate distributions. We refer to this new class of distributions as bivariate generalized linear failure rate –power series model. This new class of bivariate distributions contains several lifetime models such as: generalized linear failure rate-power series, bivariate generalized linear failure rate and bivariate generalized linear failure rate geometric distributions as special cases among others. The construction and characteristics of the proposed bivariate distribution are presented along with estimation procedures for the model parameters based on maximum likelihood. The marginal and conditional laws are also studied. We present an application to the real data set where our model provides a better fit than other models.
\end{abstract}

\section{Introduction}

Compounding continuous with discrete distributions have been introduced and studied in the recent years. These method  allow us to obtain distributions with great flexibility and are useful to develop more realistic statistical models in a great variety of applications.
\cite{ma-ol-97}
introduced a class of distributions which can be obtained by minimum and maximum of independent and identically distributed  continuous random variables, where the sample size follows geometric distribution.
\cite{si-bo-di-co-13}
introduced a class of distributions obtained by mixing extended Weibull and power series distributions and studied several of its statistical properties. This class contains the Weibull-geometric distribution
\citep{ma-ol-97,ba-de-co-11}
 and other lifetime models studied recently. The reader is referred to
 \cite{si-bo-di-co-13}
  for a brief literature review about some univariate distributions obtained by compounding.

The three-parameter generalized linear failure rate (GLFR) distribution has been introduced by
\cite{sa-ku-09}.
The hazard function of GLFR distribution can be increasing, decreasing and bathtub shaped, and it has the following cumulative distribution function (cdf) and probability density function (pdf), respectively,
\begin{eqnarray}
&&F_{\rm G}(x;\alpha)=(1-e^{-\beta x-\frac{\gamma }{2}x^2})^{\alpha },\ \ \ \ \ x>0,\ \ \ \alpha ,\beta ,\gamma >0,\\
&&f_{\rm G}(x;\alpha)=\alpha (\beta +\gamma x)e^{-\beta x-\frac{\gamma }{2}x^2}(1-e^{-\beta x-\frac{\gamma }{2}x^2})^{\alpha -1}.
\end{eqnarray}

 Recently,  many studies have been done on GLFR distribution, and some authors have extended it: the generalized linear exponential \citep{Ma-al-10}, beta-linear failure rate \citep{ja-ma-2012}, Kumaraswamy-GLFR \citep{elbatal-13}, modified-GLFR \citep{jamkhaneh-14}, McDonald-GLFR \citep{el-me-ma-14}, Poisson-GLFR \citep{co-or-le-14}, GLFR-geometric \citep{na-sh-re-14} and GLFR-power series \citep{al-sh-14} are some univariate extension of GLFR distribution.

\cite{sa-ha-sm-ku-11}
have introduced a new bivariate distribution using the GLFR distribution and derived several interesting properties of this new bivariate distribution. The proposed bivariate GLFR (BGLFR) distribution is an extension of bivariate generalized exponential (BGE) distribution \citep{ku-gu-09-BGE} and its
 cdf, the joint pdf and the joint survival distribution function are in closed forms. This bivariate distribution has both an absolute continuous part and a singular part, and is extended to a multivariate distribution. The cdf of the BGLFR model is given by
\begin{equation}\label{eq.FBGLFR}
F_{\rm BG}(x_1, x_2;\alpha_1,\alpha_2,\alpha_3)=\left\{\begin{array}{cl}
{\left(1-e^{-\beta x_1-\frac{\gamma }{2}x^2_1}\right)}^{{\alpha }_1+{\alpha }_3}{\left(1-e^{-\beta x_2-\frac{\gamma }{2}x^2_2}\right)}^{{\alpha }_2} &
\ \ {\rm if}\ \  x_1\le x_2 \\
{\left(1-e^{-\beta x_1-\frac{\gamma }{2}x^2_1}\right)}^{{\alpha }_1}{\left(1-e^{-\beta x_2-\frac{\gamma }{2}x^2_2}\right)}^{{\alpha }_2+{\alpha }_3} &
\ \ {\rm if}\ \ x_1>x_2. \end{array}
\right.
\end{equation}

In this paper, we compound the BGLFR distribution and power series class of distributions and define a new class of bivariate distributions. This class contains the BGLFR and GLFR-power series (GLFRPS) distributions and is called the bivariate GLFR-power series (BGLFRPS) class of distributions. This paper is organized as follows. In section \ref{sec.class}, we introduce the BGLFRPS distributions and obtain some properties of this new family. Some special models are studied in detail in Section \ref{sec.spe}.  An EM algorithm is proposed to estimate the model parameters   in Section \ref{sec.est}. A real data application of the BGLFRPS distributions is illustrated in Section \ref{sec.exa}.

\section{The BGLFRPS class}

\label{sec.class}
Let the random variable $N$ be a discrete random variable having
 a power series distribution  (truncated at zero) with probability mass function (pmf)
\begin{equation}
P(N=n)=‎\frac{a_n \theta ^{n} }{C(\theta)}, \ \ \ n=1,2,‎‎\ldots‎‎ ,‎
\end{equation}
where $a_n‎\geq 0‎$ depends only on $n$, $C(\theta)=\sum_{n=1}^{\infty} a_n \theta ^{n} $ and $\theta ‎\in (0,s)‎$ ($s$ can be $\infty$) is such that $C(\theta)$ is finite.
 Table \ref{tab.ps}
 summarizes some particular cases of the truncated (at zero) power series distributions (geometric, Poisson, logarithmic, binomial and negative binomial). Detailed properties of power series distribution can be found in
\cite{noack-50}.
Here, $C'(\theta )$, $C''(\theta )$ and $C'''(\theta )$ denote the first, second and third derivatives of $C(\theta)$ with respect to $\theta$, respectively.

\begin{table}[ht]

\begin{center}
\caption{Useful quantities for some power series distributions.}\label{tab.ps}
\begin{tabular}{|l| c c c c c c| }\hline
‎ Distribution
& $a_n$ & $C(\theta)$ & $C^{\prime}(‎\theta)$ & $C^{\prime\prime}(‎\theta)$ & $C^{\prime\prime\prime}(‎\theta)$ &   $s$ \\ \hline
Geometric & $1$ & $‎\theta (1-‎\theta)^{-1}$ & $(1-‎\theta)^{-2}$ & $2(1-‎\theta)^{-3}$ & $6(1-‎\theta)^{-4}$ &  $1$ \\
Poisson & $n!^{-1}$ & $e^{‎\theta}-1$ & $e^{‎\theta}$ & $e^{‎\theta}$ & $e^{‎\theta}$ &  $\infty‎$ \\
Logarithmic & $n^{-1}$ & $-\log(1-‎\theta)$ & $(1-‎\theta)^{-1}$ & $(1-‎\theta)^{-2}$ & $2(1-‎\theta)^{-3}$  & $1$ \\
Binomial & $\binom {k} {n}$
& $(1+‎\theta)^k-1$ & $\frac{k}{(\theta+1)^{1-k}}‎$ & $\frac{k(k-1)}{(\theta+1)^{2-k}}‎$ & $\frac{k(k-1)(k-2)}{(\theta+1)^{3-k}}$ &  $\infty$ \\
Negative Binomial & $\binom {n-1} {k-1}$
& $‎\frac{\theta^k}{ (1-‎\theta)^{k}}$ & $‎\frac{k\theta^{k-1}}{ (1-‎\theta)^{k+1}}$ & $‎\frac{k(k+2\theta-1)}{\theta^{2-k} (1-‎\theta)^{k+2}}$&
$‎\frac{k(k^2+6k\theta+6\theta^2-3k-6\theta+2)}{ \theta^{3-k}(1-‎\theta)^{k+3}}$
& $1$  \\ \hline
\end{tabular}

\end{center}
\end{table}

Now, suppose $\{(X_{1n}, X_{2n}); n=1,2\dots\}$ is a sequence of independent and identically distributed non-negative bivariate random variables with common joint cdf $F_{\boldsymbol X}(.,.)$, where ${\boldsymbol X}=(X_1, X_2)'$. Take $N$ to be a power series random variable independent of $(X_{1i}, X_{2i})$. Let
\[Y_i=\max  \{X_{i1},\dots X_{iN}\} ,\ \ \ \ i=1,2.\]
For given $N=n$,  the joint cdf of ${\boldsymbol Y}=(Y_1,Y_2)$ is
\begin{equation}\label{eq.pdfyn}
F_{Y_1,Y_2|N}(y_1, y_2|n)=(F_{{\boldsymbol X}}(y_1, y_2))^n.
\end{equation}
Therefore, the joint cdf of ${\boldsymbol Y}$ becomes
\begin{equation}
F_{{\boldsymbol Y}}(y_1, y_2)=\sum^{\infty }_{n=1}{{(F_{{\boldsymbol X}}(y_1,y_2))}^n\frac{a_n{\theta }^n}{C(\theta )}}=\frac{C(\theta F_{{\boldsymbol X}}(y_1, y_2))}{C(\theta )}.
\end{equation}
In this case, we call ${\boldsymbol Y}$ has a bivariate F-power series (BFPS) distribution.

The corresponding marginal cdf of $Y_i$ is
\[F_{Y_i}(y_i)=\frac{C(\theta F_{{{\boldsymbol X}}_{{\boldsymbol i}}}(y_i))}{C(\theta )},\ \ \ \ \ i=1,2, \]
and  recently this univariate class is considered
by many authors: for example the
generalized exponential-power series \citep{ma-ja-12},
complementary exponential-power series
\citep{fl-bo-ca-13},
 complementary extended Weibull-power series
\citep{co-si-14} and
and GLFRPS \citep{al-sh-14}
distributions.

In this paper, we take $F$ to be the bivariate GLFR distribution  given in \eqref{eq.FBGLFR}.
Therefore, we consider the BGLFRPS class of distributions which is defined by the following cdf:
\begin{eqnarray}\label{eq.FBGLFRPS}
F_{{\boldsymbol Y}}(y_1,y_2)&=&
\left\{ \begin{array}{cl}
\frac{C\left(\theta (1-e^{-\beta y_1-\frac{\gamma }{2}y^2_1})^{{\alpha }_1+{\alpha }_3}(1-e^{-\beta y_2-\frac{\gamma }{2}y^2_2})^{{\alpha }_2}\right)}{C(\theta )} & \ \ {\rm if}\ \  y_1\le y_2 \\
\frac{C\left(\theta (1-e^{-\beta y_1-\frac{\gamma }{2}y^2_1})^{{\alpha }_1}(1-e^{-\beta y_2-\frac{\gamma }{2}y^2_2})^{{\alpha }_2+{\alpha }_3}\right)}{C\left(\theta \right)} &
\ \ {\rm if}\ \ y_1>y_2, \end{array}
\right.
 \nonumber\\
&=&\left\{ \begin{array}{cl}
\frac{C\left(\theta F_{{\rm G}}(y_1;{\alpha }_1+{\alpha }_3 )F_{{\rm G}}(y_2;{\alpha }_2)\right)}{C(\theta )} &
 \ \ \ \ {\rm if}\ \ y_1\le y_2 \\
\frac{C\left(\theta F_{{\rm G}}(y_1;{\alpha }_1)F_{\rm G}(y_2;{\alpha }_2+{\alpha }_3)\right)}{C(\theta)} &
\ \ \ \ {\rm if}\ \ y_1>y_2, \end{array}
\right.
\end{eqnarray}
and is denoted by$\ {\rm B}{\rm GLFR}{\rm PS}\left({\alpha }_1,{\alpha }_2,{\alpha }_3,\beta ,\gamma ,\theta \right)$.

\begin{theorem} \label{thm.pdf}
Let ${\boldsymbol Y}=\left(Y_1,Y_2\right)$ has a ${\rm BGLFRPS}\left({\alpha }_1,{\alpha }_2,{\alpha }_3,\beta ,\gamma ,\theta \right)$ distributions. Then the joint pdf of ${\boldsymbol Y}$ is
\begin{equation}\label{eq.fBGLFRPS}
f_{{\boldsymbol Y}}\left(y_1,y_2\right)=\left\{ \begin{array}{ll}
f_1\left(y_1,y_2\right) &  \ \ {\rm if}\ \  0<y_1<y_2 \\
f_2\left(y_1,y_2\right) & \ \ {\rm if}\ \  0<y_2<y_1 \\
f_0\left(y\right)  &  \ \ {\rm if}\ \ 0<y_1=y_2=y, \end{array}
\right.
\end{equation}
where
\begin{eqnarray}
f_1(y_1,y_2)&=&\frac{\theta }{C(\theta)}f_{\rm G}(y_1;{\alpha }_1+{\alpha }_3)f_{\rm G}(y_2;{\alpha }_2)
\left[\theta F_{\rm G}(y_1;{\alpha}_1+{\alpha }_3)F_{\rm G}(y_2;{\alpha }_2)\right. \label{eq.f1} \\
&&\times C''(\theta F_{\rm G}(y_1;{\alpha }_1+{\alpha }_3)F_{\rm G}(y_2;{\alpha }_2))
+\left.C'(\theta F_{\rm G}(y_1;{\alpha }_1+{\alpha }_3)F_{\rm G}(y_2;{\alpha }_2))\right], \nonumber \\
f_2(y_1,y_2)&=&\frac{\theta }{C(\theta)}f_{{\rm G}}(y_1;{\alpha }_1 )f_{{\rm G}}(y_2;{\alpha }_2+{\alpha }_3)
\left[\theta F_{{\rm G}}(y_1;{\alpha }_1)F_{\rm G}(y_2;{\alpha }_2+{\alpha }_3 )\right. \label{eq.f2} \\
&&\times C''(\theta F_{{\rm G}}(y_1;{\alpha }_1)F_{{\rm G}}(y_2;{\alpha }_2+{\alpha }_3 ))\left.+C'(\theta F_{{\rm G}}(y_1;{\alpha }_1)F_{{\rm G}}(y_2;{\alpha }_2+{\alpha }_3))\right],\nonumber \\
f_0(y)&=&\frac{\theta {\alpha }_3}{C(\theta )({\alpha }_1+{\alpha }_2+{\alpha }_3)}
f_{{\rm G}}(y;{\alpha }_1+{\alpha }_2+{\alpha }_3)C'(\theta F_{{\rm G}}(y;{\alpha }_1+{\alpha }_2+{\alpha }_3)).\label{eq.f0}
\end{eqnarray}

\end{theorem}

\begin{proof}
It is obvious.
\end{proof}

As a special case, we consider $C(\theta)=\theta +{\theta }^{20}$ \citep[see also][]{ma-ja-12,mo-ba-11}.
The pdf of the BGLFRPS class of distributions are depicted in Figure \ref{fig.denh} for $\beta =\gamma=1$ and some values of other parameters.

\begin{figure}[ht]
\centering
\includegraphics[width=7cm,height=6cm]{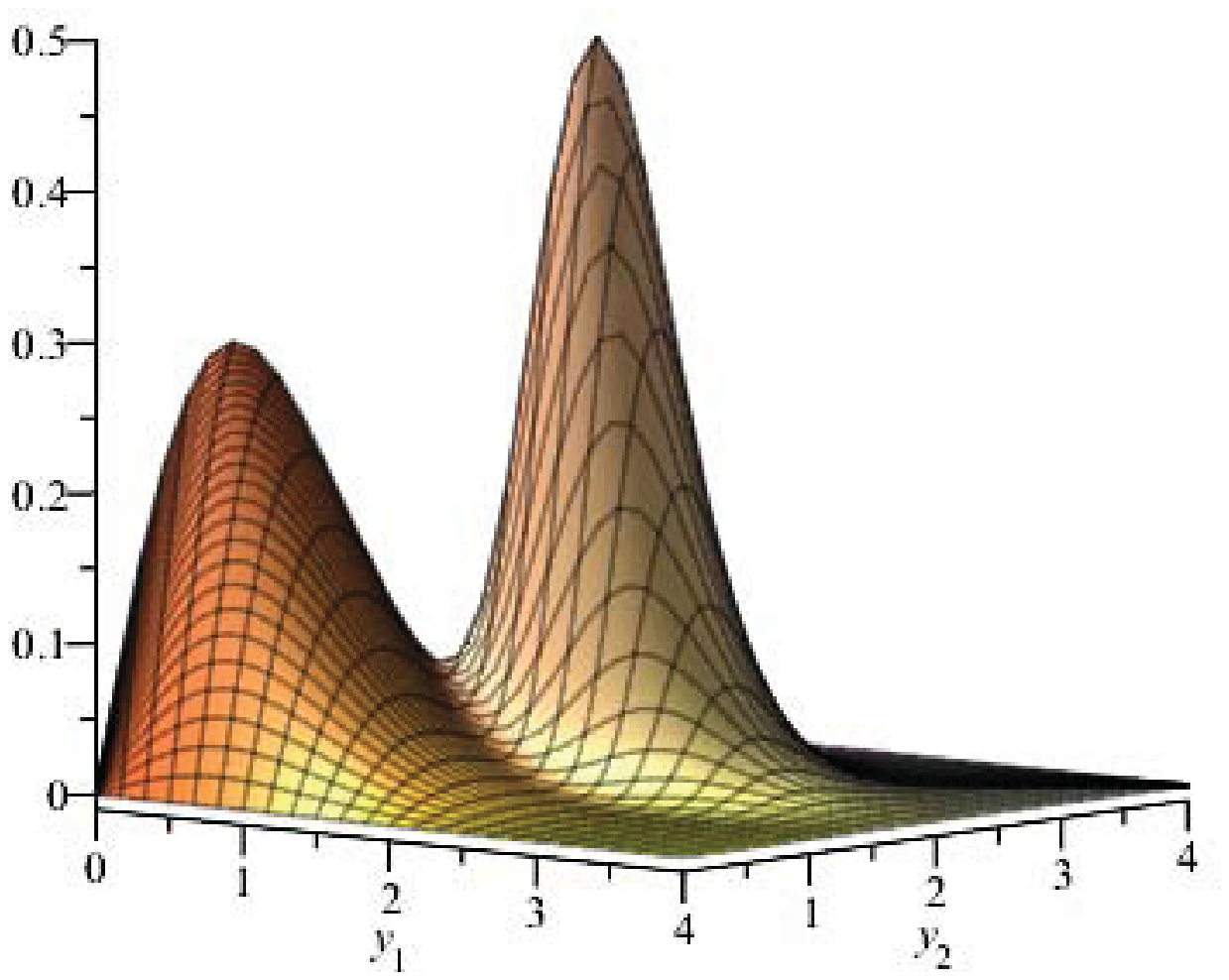}
\includegraphics[width=7cm,height=6cm]{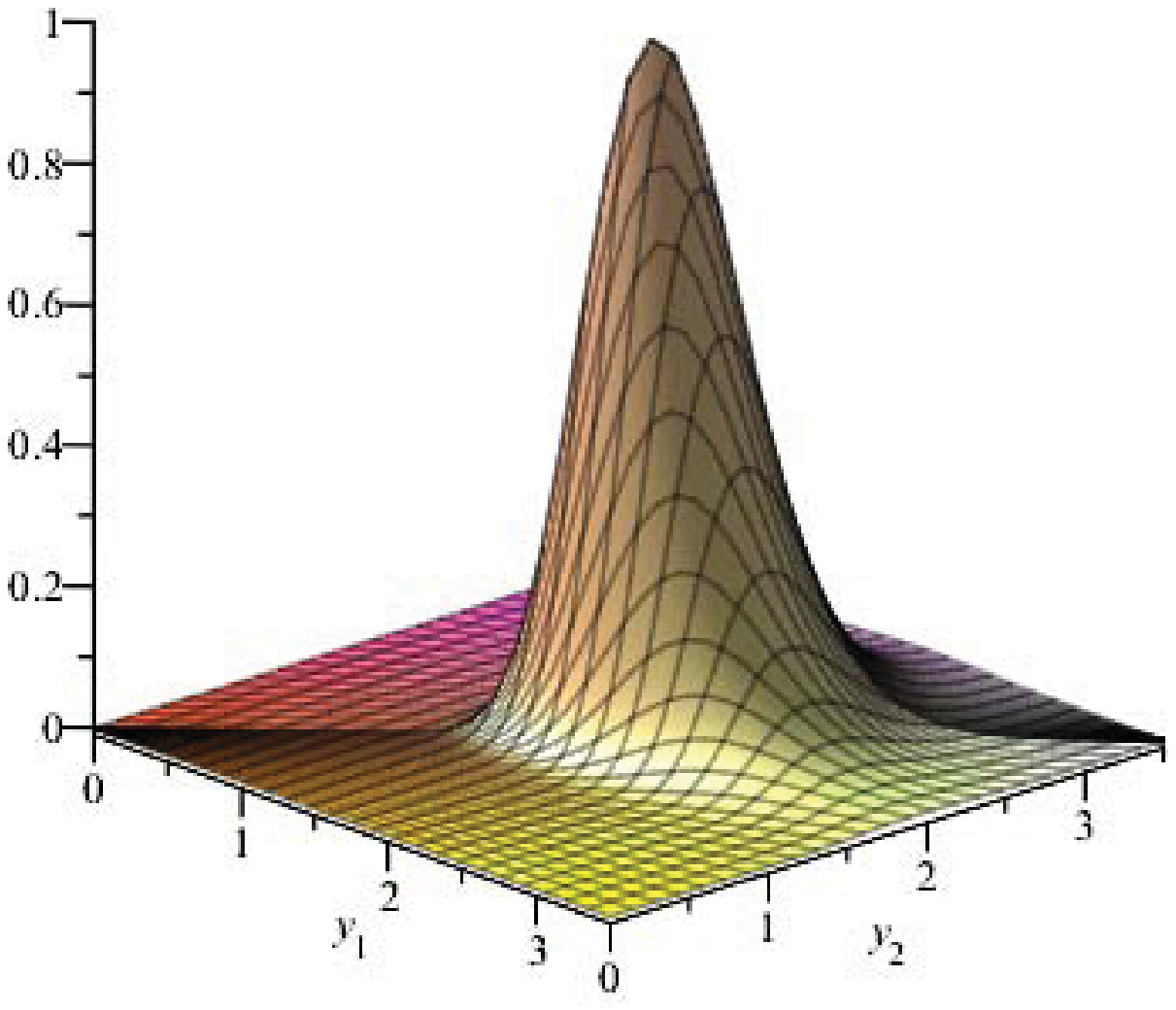}
\includegraphics[width=7cm,height=6cm]{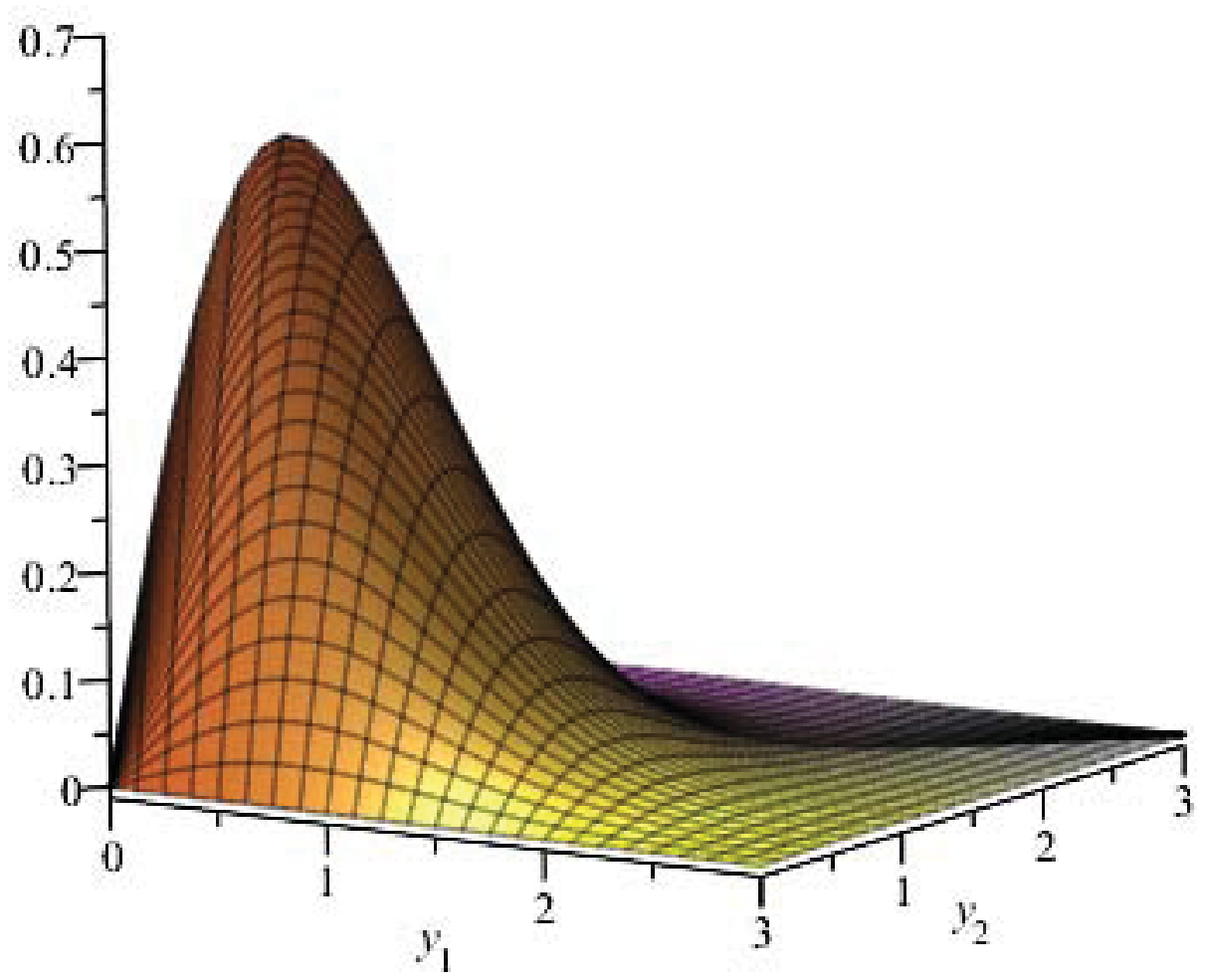}
\includegraphics[width=7cm,height=6cm]{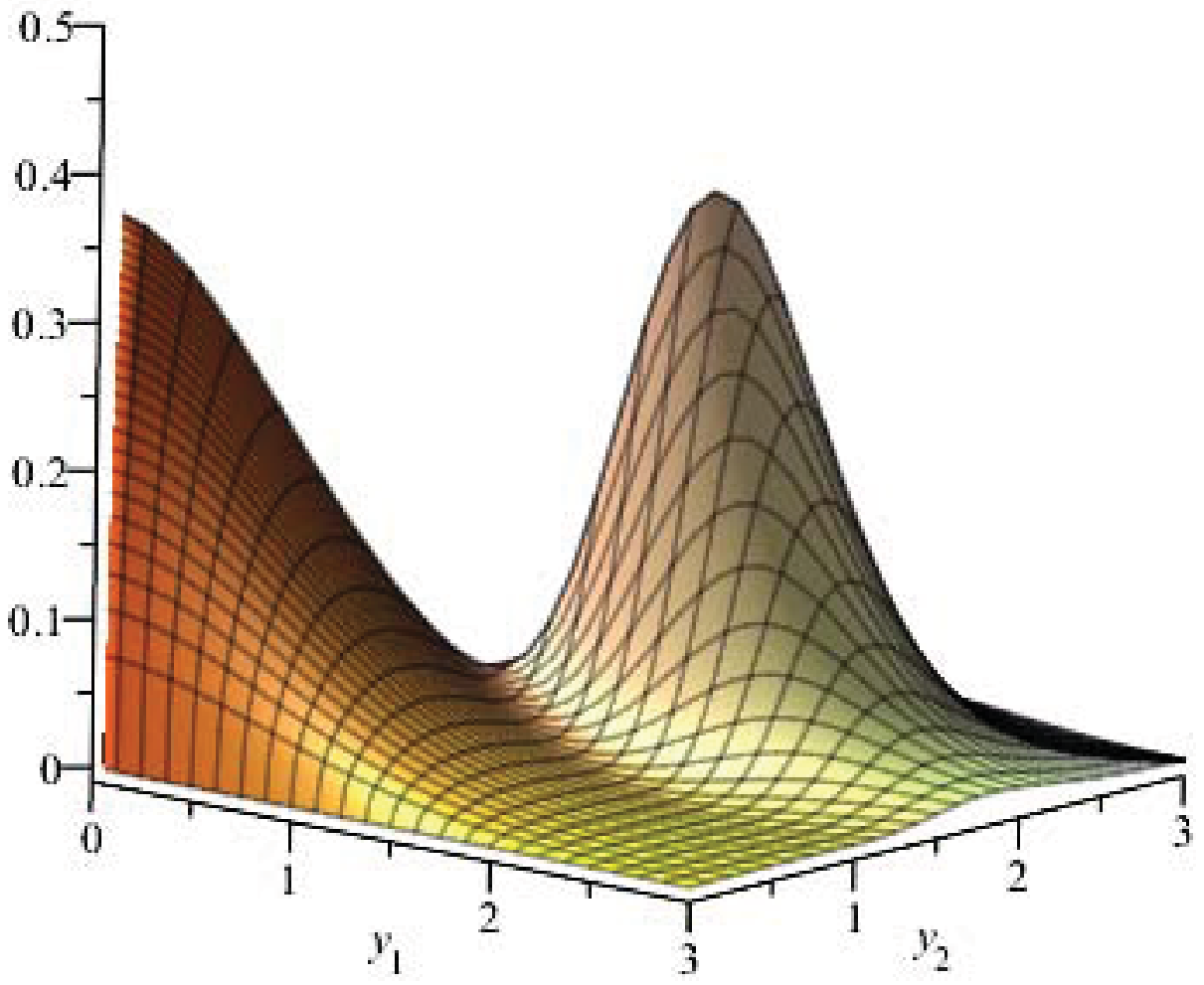}
\caption[]{The pdf of the BGLFRPS class of distribution for some values of parameters:
${\alpha }_1={\alpha }_2={\alpha }_3=1$, $\theta =1$ (left top),
${\alpha }_1={\alpha }_2={\alpha }_3=1$, $\theta =2$ (right top),
${\alpha }_1={\alpha }_2={\alpha }_3=1$, $\theta =0.5$ (left bottom),
${\alpha }_1={\alpha }_2=0.5, {\alpha }_3=1$, $\theta =1$ (right bottom).}\label{fig.denh}
\end{figure}

\begin{prop}
Let $(Y_1,Y_2)$ has a ${\rm BGLFRPS}\left({\alpha }_1,{\alpha }_2,{\alpha }_3,\beta,\gamma ,\theta \right)$ distributions
 Then\\
1. Each $Y_i$ has a GLFRPS distributions with parameters ${\alpha }_i+{\alpha }_3$, $\beta,\gamma $ and $\theta $ with following cdf:
$$
F(x)=\frac{C\left(\theta (1-e^{-\beta x-\frac{\gamma }{2}x})^{{\alpha }_i+{\alpha }_3}\right)}{C(\theta )}, \ \ \ x>0.
$$
2. The random variable $U={\max  (Y_1,Y_2)}$ has a GLFRPS distributions with parameters ${\alpha }_1+{\alpha }_2+{\alpha }_3$, $\beta ,\gamma $ and $\theta $.\\
3. If $C(\theta )=\theta $, then ${\boldsymbol Y}$ has a BGLFR distribution with parameters
${\alpha }_1$, ${\alpha }_2$, ${\alpha }_3$, $\beta $ and $\gamma $.\\
4. $P(Y_1<Y_2)=\frac{{\alpha }_1}{{\alpha }_1+{\alpha }_2+{\alpha }_3}$.
\end{prop}

\begin{prop} Let $F_{{\boldsymbol Y}}(y_1,y_2)$ be the cdf of BGLFRPS distributions given in \eqref{eq.FBGLFRPS}. Then
\[F_{{\boldsymbol Y}}(y_1,y_2)=\sum^{\infty }_{n=1}{p_nF_{{\rm BG}}(y_1,y_2;n{\alpha }_1,n{\alpha }_2,n{\alpha }_3),}\]
where $p_n=P\left(N=n\right)=\frac{a_n{\theta }^n}{C\left(\theta \right)}$. Therefore
\begin{eqnarray*}
&&f_1(y_1,y_2)=\sum^{\infty }_{n=1}{p_nf_{{\rm G}}\left(y_1;n{\alpha }_1+n{\alpha }_3\right)f_{{\rm G}}\left(y_2;n{\alpha }_2\right)},\\
&&f_2(y_1,y_2)=\sum^{\infty }_{n=1}{p_nf_{{\rm G}}\left(y_1;n{\alpha }_1\right)f_{{\rm G}}\left(y_2;n{\alpha }_2+n{\alpha }_3\right)},\\
&&f_0(y)=\frac{{\alpha }_3}{{\alpha }_1+{\alpha }_2+{\alpha }_3}\sum^{\infty }_{n=1}{p_nf_{{\rm G}}\left(y;n{\alpha }_1+n{\alpha }_2+n{\alpha }_3\right)},
\end{eqnarray*}
where $f_{{\rm G}}\left(.;n\alpha \right)$ is the pdf of GLFR distribution with parameters $n\alpha $, $\beta $ and $\gamma $.
Note that $f_{{\rm G}}\left(.;n\alpha \right)$ is the pdf of random variable ${\max  \left(U_1,\dots ,U_n\right)}$ where $U_i$'s are independent random variables from a GLFR distribution with parameters $\alpha $, $\beta $ and $\gamma $.
\end{prop}

\begin{prop} The joint pdf of the BGLFRPS distributions provided in Theorem \ref{thm.pdf} can be written as
\[f_{{\boldsymbol Y}}(y_1,y_2)=\frac{{\alpha }_1+{\alpha }_2}{{\alpha }_1+{\alpha }_2+{\alpha }_3}g_a(y_1,y_2)+\frac{{\alpha }_3}{{\alpha }_1+{\alpha }_2+{\alpha }_3}g_s(y),\]
where
\[g_a(y_1,y_2)=\frac{{\alpha }_1+{\alpha }_2+{\alpha }_3}{{\alpha }_1+{\alpha }_2}
\left\{ \begin{array}{ll}
f_1\left(y_1,y_2\right) & \ \ {\rm if}\ \  y_1<y_2 \\
f_2\left(y_1,y_2\right) & \ \ {\rm if}\ \ y_2<y_1,
\end{array}
\right.\]
\[g_s(y)=\frac{\theta }{C(\theta)}f_{{\rm G}}(y;{\alpha }_1+{\alpha }_2+{\alpha }_3)
C'(\theta F_{{\rm G}}(y;{\alpha }_1+{\alpha }_2+{\alpha }_3)){\rm \ \ if \ \ }y_1=y_2=y,\]
and 0 otherwise. Clearly $g_a(.,.)$ is the absolute continuous part and  $g_s(.)$ is the singular part. If ${\alpha }_3=0,$ it does not have any singular part and it becomes an absolute continuous pdf.
\end{prop}

\begin{prop}
The conditional distribution of $Y_1$ given $Y_2\le y_2$ is an absolute continuous distribution with the following cdf:
\[P(Y_1\le y_1|Y_2\le y_2)=
\left\{ \begin{array}{ll}
\frac{C\left(\theta
(1-e^{-\beta y_1-\frac{\gamma }{2}y^2_1})^{{\alpha }_1+{\alpha }_3}
(1-e^{-\beta y_2-\frac{\gamma }{2}y^2_2})^{{\alpha }_2}\right)}{
C\left(\theta {(1-e^{-\beta y_2-\frac{\gamma }{2}y^2_2})}^{{\alpha }_2+{\alpha }_3}\right)}
& \ \ {\rm if}\ \  y_1<y_2 \\
\frac{C\left(\theta {(1-e^{-\beta y_2-\frac{\gamma }{2}y^2_2})}^{{\alpha }_2+{\alpha }_3}
{(1-e^{-\beta y_1-\frac{\gamma }{2}y^2_1})}^{{\alpha }_1}\right)}{C\left(\theta {(1-e^{-\beta y_2-\frac{\gamma }{2}y^2_2})}^{{\alpha }_2+{\alpha }_3}\right)}
& \ \ {\rm if}\ \ y_2<y_1.
\end{array}
\right.\]
\end{prop}

\begin{prop} The limiting distribution of BGLFRPS when $\theta \to 0^+$ is
\begin{eqnarray*}
{\mathop{\lim }_{\theta \to 0^+} F_{{\boldsymbol Y}}(y_1,y_2) }
&=&{\mathop{\lim }_{\theta \to 0^+} \frac{C\left(\theta F_{{\boldsymbol X}}(y_1,y_2)\right)}{C(\theta )} }\\
&=&{\mathop{\lim}_{\theta \to 0^+} \frac{\sum^{\infty }_{n=1}{a_n{\theta }^n{(F_{{\boldsymbol X}}(y_1,y_2))}^n}}{\sum^{\infty }_{n=1}{a_n{\theta }^n}}}\\
&=&{\mathop{\lim }_{\theta \to 0^+} \frac{a_c{(F_{{\boldsymbol X}}(y_1, y_2))}^c+\sum^{\infty }_{n=c+1}{a_n{\theta }^{n-c}{(F_{{\boldsymbol X}}
(y_1,y_2))}^n}}{a_c+\sum^{\infty }_{n=c+1}{a_n{\theta }^{n-c}}}}\\
&=&{\left(F_{{\boldsymbol X}}(y_1,y_2)\right)}^c\\
&=&\left\{ \begin{array}{ll}
{(1-e^{-\beta y_1-\frac{\gamma }{2}y^2_1})}^{c({\alpha }_1+{\alpha }_3)}{(1-e^{-\beta y_2-\frac{\gamma }{2}y^2_2})}^{c{\alpha }_2}
& \ \ {\rm if}\ \  y_1\le y_2 \\
{\left(1-e^{-\beta y_1-\frac{\gamma }{2}y^2_1}\right)}^{c{\alpha }_1}{\left(1-e^{-\beta y_2-\frac{\gamma }{2}y^2_2}\right)}^{c({\alpha }_2+{\alpha }_3)} &
\ \ {\rm if}\ \ y_1>y_2, \end{array}
\right.
\end{eqnarray*}
which is the pdf of a BGLFR distribution with parameters $c{\alpha }_1$, $c{\alpha }_2$, $c{\alpha }_3$, $\beta $ and $\gamma $, where
$c=\min  \left\{n\in {\mathbb N} :a_n>0\right\}$.
\end{prop}

Based on \eqref{eq.pdfyn}, $(Y_1,Y_2)$ given $N=n$ has a BGLFR with parameters $n{\alpha }_1$, $n{\alpha }_2$, $n{\alpha }_3$, $\beta $ and $\gamma$, and therefore, the joint pdf of $(Y_1,Y_2,N)$ is
\[f_{Y_1,Y_2,N}(y_1, y_2,n)=
\left\{ \begin{array}{ll}
\frac{a_n{\theta }^n}{C(\theta )}f_{1n}(y_1,y_2) & \ \ {\rm if} \ \ y_1<y_2 \\
\frac{a_n{\theta }^n}{C(\theta )}f_{2n}(y_1,y_2) & \ \ {\rm if} \ \ y_2<y_1 \\
\frac{a_n{\theta }^n}{C(\theta )}f_{0n}(y) & \ \ {\rm if} \ \ y_1=y_2=y,
\end{array}
\right.\]
where
\begin{eqnarray*}
f_{1n}(y_1,y_2)&=&n^2(\beta +\gamma y_1)(\beta +\gamma y_2)({\alpha }_1+{\alpha }_3){\alpha }_2e^{-\beta {(y}_1+y_2)-\frac{\gamma }{2}{(y}^2_1+y^2_2)}\\
&&\times(1-e^{-\beta y_1-\frac{\gamma }{2}y^2_1})^{n({\alpha }_1+{\alpha }_3)-1}(1-e^{-\beta y_2-\frac{\gamma }{2}y^2_2})^{n{\alpha }_2-1},\\
f_{2n}(y_1,y_2)&=&n^2(\beta +\gamma y_1)(\beta +\gamma y_2)({\alpha }_2+{\alpha }_3){\alpha }_1e^{-\beta {(y}_1+y_2)-\frac{\gamma }{2}{(y}^2_1+y^2_2)}\\
&&\times(1-e^{-\beta y_1-\frac{\gamma }{2}y^2_1})^{n{\alpha }_1-1}(1-e^{-\beta y_2-\frac{\gamma }{2}y^2_2})^{n({\alpha }_2+{\alpha }_3)-1},\\
f_{0n}(y)&=&n(\beta +\gamma y){\alpha }_3e^{-\beta y-\frac{\gamma }{2}y^2}(1-e^{-\beta y-\frac{\gamma }{2}y^2})^{n({\alpha }_1+{\alpha }_2+{\alpha }_3)-1}.
\end{eqnarray*}
The conditional pmf of $N$ given $Y_1=y_1$ and $Y_2=y_2$ is
\begin{equation}\label{eq.fNY}
f_{N|Y_1,Y_2}(n|y_1,y_2)=\left\{
\begin{array}{ll}
\frac{n^2a_n{\left(\theta  A_1(y_1,y_2)\right)}^{n-1}}{k_1(y_1,y_2)} & \ \ {\rm if} \ \ y_1<y_2 \\
\frac{n^2a_n{\left(\theta  A_2(y_1,y_2)\right)}^{n-1}}{k_2(y_1,y_2)} & \ \ {\rm if} \ \ y_2<y_1 \\
\frac{na_n{(\theta  A_0(y))}^{n-1}}{k_0(y)} & \ \ {\rm if} \ \ y_1=y_2=y, \end{array}
\right.
\end{equation}
where
\begin{eqnarray*}
k_1(y_1,y_2)&=&\theta F_{{\rm G}}(y_1;{\alpha }_1+{\alpha }_3)F_{{\rm G}}(y_2;{\alpha }_2)
C''(\theta F_{{\rm G}}(y_1;{\alpha }_1+{\alpha }_3)F_{{\rm G}}(y_2;{\alpha }_2))\\
&&+C'(\theta F_{{\rm G}}(y_1;{\alpha }_1+{\alpha }_3)F_{{\rm G}}(y_2;{\alpha }_2)),\\
k_2(y_1,y_2)&=&\theta F_{{\rm G}}(y_1;{\alpha }_1)F_{{\rm G}}(y_2;{\alpha }_2+{\alpha }_3)C''(\theta F_G(y_1;{\alpha }_1)F_{{\rm G}}(y_2;{\alpha }_2+{\alpha }_3))\\
&&+C'(\theta F_{{\rm G}}(y_1;{\alpha }_1)F_{{\rm G}}(y_2;{\alpha }_2+{\alpha }_3)),\\
k_0(y)&=&C'(\theta F_{{\rm G}}(y;{\alpha }_1+{\alpha }_2+{\alpha }_3)),
\end{eqnarray*}
and
\begin{eqnarray*}
&&A_1\left(y_1,y_2\right)=(1-e^{-\beta y_1-\frac{\gamma }{2}y^2_1})^{{\alpha }_1+{\alpha }_3}(1-e^{-\beta y_2-\frac{\gamma }{2}y^2_2})^{{\alpha }_2}=F_{{\rm G}}\left(y_1;{\alpha }_1+{\alpha }_3\right)F_{{\rm G}}\left(y_2;{\alpha }_2\right)\\
&&A_2\left(y_1,y_2\right)=(1-e^{-\beta y_1-\frac{\gamma }{2}y^2_1})^{{\alpha }_1}(1-e^{-\beta y_2-\frac{\gamma }{2}y^2_2})^{{\alpha }_2+{\alpha }_3}=F_{{\rm G}}\left(y_1;{\alpha }_1\right)F_{{\rm G}}\left(y_2;{\alpha }_2+{\alpha }_3\right)\\
&&A_0\left(y\right)=(1-e^{-\beta y-\frac{\gamma }{2}y^2})^{{\alpha }_1+{\alpha }_2+{\alpha }_3}{=F}_{{\rm G}}\left(y;{{\alpha }_1+\alpha }_2+{\alpha }_3\right)
\end{eqnarray*}

Since
$\theta^2C'''(\theta)+3\theta C''\left(\theta \right)+C'\left(\theta \right)=\sum^{\infty }_{n=1}{n^3a_n{\theta }^{n-1}}$,
$\theta C''(\theta)+C'\left(\theta \right)=\sum^{\infty }_{n=1}{n^2a_n{\theta }^{n-1}}$
and $C'\left(\theta \right)=\sum^{\infty }_{n=1}{na_n{\theta }^{n-1}}$,  we can obtain the conditional expectation of $N$ given $Y_1=y_1$ and $Y_2=y_2$ as
\begin{equation}\label{eq.ENY}
E(N| Y_1,Y_2)=\left\{
\begin{array}{ll}
\frac{B_1(y_1,y_2)}{k_1(y_1,y_2)} & \ \ {\rm if} \ \ y_1<y_2 \\
\frac{B_2\left(y_1,y_2\right)}{k_2(y_1,y_2)} & \ \ {\rm if} \ \ y_2>y_1 \\
\frac{\theta A_0(y)C''\left(\theta A_0\left(y\right)\right)+C'\left(\theta A_0\left(y\right)\right)}{k_0(y)} & \ \ {\rm if} \ \ y_1=y_2=y. \end{array}
\right.
\end{equation}
where
\begin{eqnarray*}
B_i(y_1,y_2)&=&(\theta A_i(y_1,y_2))^2C'''(\theta A_i(y_1,y_2))\\
&&+3\theta A_i(y_1,y_2)C''(\theta A_i(y_1,y_2))\\
&&+C'(\theta A_i(y_1,y_2)),\ \ \ \ \ \ \ \ \ \ \ \ \ \ \ \ \ \ \  i=1,2.
\end{eqnarray*}

\begin{remark} If we consider $Z_i=\min  \left\{X_{i1},\dots X_{iN}\right\}$, $ i=1,2,$
another class of bivariate distribution is obtained with the following joint cumulative survival function:
\[{\bar{F}}_{{Z_1,Z_2}}(z_1,z_2)=P(Z_1>z_1,Z_2>z_2)=\frac{C(\theta {\bar{F}}_{{\boldsymbol X}}(z_1, z_2))}{C(\theta)},\]
where ${\bar{F}}_{{\boldsymbol X}}(x_1, x_2)=P(X_1>x_1,X_2>x_2)$. This univariate class of distributions is studied in literature:
for example the exponential-power series
\citep{ch-ga-09},
and the Weibutll-power series
\citep{mo-ba-11},
Burr XII power series
\citep{si-co-13},
 double bounded Kumaraswamy-power series
\citep{bi-ne-13},
Birnbaum-Saunders-power series
\citep{bo-si-co-14} and
linear failure rate-power series
\citep{ma-ja-14}
distributions.
\end{remark}

\section{Special Cases}
\label{sec.spe}
In this section, we consider some special cases of BGLFRPS distributions.

\subsection{Bivariate GLFR-geometric distribution}

When $C\left(\theta \right)=\frac{\theta }{1-\theta }$ ($0<\theta <1$), the power series distribution becomes the geometric distribution (truncated at zero). Therefore, the cdf of bivariate GLFR-geometric (BGLFRG) distribution is given by
\[F_{{\boldsymbol Y}}(y_1,y_2)=\left\{
\begin{array}{ll}
\frac{(1-\theta ){(1-e^{-\beta y_1-\frac{\gamma }{2}y^2_1})}^{{\alpha }_1+{\alpha }_3}{(1-e^{-\beta y_2-\frac{\gamma }{2}y^2_2})}^{{\alpha }_2}}{1-\theta {(1-e^{-\beta y_1-\frac{\gamma }{2}y^2_1})}^{{\alpha }_1+{\alpha }_3}{(1-e^{-\beta y_2-\frac{\gamma }{2}y^2_2})}^{{\alpha }_2}}
& \ \ {\rm if}\ \   y_1\le y_2 \\
\frac{(1-\theta ){(1-e^{-\beta y_1-\frac{\gamma }{2}y^2_1})}^{{\alpha }_1}{(1-e^{-\beta y_2-\frac{\gamma }{2}y^2_2})}^{{\alpha }_2+{\alpha }_3}}{1-\theta {(1-e^{-\beta y_1-\frac{\gamma }{2}y^2_1})}^{{\alpha }_1}{(1-e^{-\beta y_2-\frac{\gamma }{2}y^2_2})}^{{\alpha }_2+{\alpha }_3}} & \ \ {\rm if}\ \ y_1>y_2,
\end{array}
\right.\]
and its pdf is given in \eqref{eq.fBGLFRPS} with
\begin{eqnarray*}
f_1(y_1,y_2)&=&(1-\theta )f_{\rm G}(y_1;{\alpha }_1+{\alpha }_3)f_{\rm G}(y_2;{\alpha }_2)\frac{1+\theta F_{\rm G}(y_1;{\alpha }_1
+{\alpha }_3)F_{\rm G}(y_2;{\alpha }_2 )}{{(1-\theta F_{\rm G}(y_1;{\alpha }_1+{\alpha }_3)F_{\rm G}(y_2;{\alpha }_2))}^3},\\
f_2(y_1,y_2)&=&(1-\theta )f_{\rm G}(y_1;{\alpha }_1 )f_{\rm G}(y_2;{\alpha }_2+{\alpha }_3)\frac{1+\theta F_{\rm G}
(y_1;{\alpha }_1)F_{\rm G}(y_2;{\alpha }_2+{\alpha }_3)}{{(1-\theta F_{\rm G}(y_1;{\alpha }_1)F_{\rm G}(y_2;{\alpha }_2+{\alpha }_3 ))}^3},\\
f_0(y)&=&\frac{(1-\theta ){\alpha }_3f_{\rm G}(y;{\alpha }_1+{\alpha }_2+{\alpha }_3)}{({\alpha }_1+{\alpha }_2+{\alpha }_3)
{(1-\theta F_{\rm G}(y;{\alpha }_1+{\alpha }_2+{\alpha }_3))}^2}.
\end{eqnarray*}

\begin{remark}  When ${\theta }^*=1-\theta $, we have
\[F_{{\boldsymbol Y}}(y_1,y_2)=\left\{ \begin{array}{ll}
\frac{{\theta }^*{(1-e^{-\beta y_1-\frac{\gamma }{2}y^2_1})}^{{\alpha }_1+{\alpha }_3}{(1-e^{-\beta y_2-\frac{\gamma }{2}y^2_2})}^{{\alpha }_2}}{1-(1-{\theta }^*){(1-e^{-\beta y_1-\frac{\gamma }{2}y^2_1})}^{{\alpha }_1+{\alpha }_3}{(1-e^{-\beta y_2-\frac{\gamma }{2}y^2_2})}^{{\alpha }_2}} & \ \ {\rm if}\ \ y_1\le y_2 \\
\frac{{\theta }^*{(1-e^{-\beta y_1-\frac{\gamma }{2}y^2_1})}^{{\alpha }_1}{(1-e^{-\beta y_2-\frac{\gamma }{2}y^2_2})}^{{\alpha }_2+{\alpha }_3}}{1-(1-{\theta }^*){(1-e^{-\beta y_1-\frac{\gamma }{2}y^2_1})}^{{\alpha }_1}{(1-e^{-\beta y_2-\frac{\gamma }{2}y^2_2})}^{{\alpha }_2+{\alpha }_3}} & \ \ {\rm if}\ \ y_1>y_2. \end{array}
\right.\]
It is also the cdf for all ${\theta }^*>0$
\citep[see][]{ma-ol-97}.
 In fact, this is in Marshal-Olkin bivariate class of distributions. Also, the marginal distribution of $Y_i$ is GLFR-geometric distribution introduced by \cite{na-sh-re-14}.
\end{remark}

\subsection{Bivariate GLFR-Poisson distribution}
When $a_n=\frac{1}{n!}$ and $C\left(\theta \right)=e^{\theta }-1$ ($\theta >0$), the power series distribution becomes the Poisson distribution (truncated at zero). Therefore, the cdf of bivariate GLFR-Poisson (BGLFRP) distribution is given by
\[F_{{\boldsymbol Y}}\left(y_1,y_2\right)=\left\{
\begin{array}{ll}
\frac{\exp\left\{\theta {(1-e^{-\beta y_1-\frac{\gamma }{2}y^2_1})}^{{\alpha }_1+{\alpha }_3}{(1-e^{-\beta y_2-\frac{\gamma }{2}y^2_2})}^{{\alpha }_2}\right\}-1}{e^{\theta }-1}
& \ \ {\rm if}\ \  y_1\le y_2 \\
\frac{\exp\left\{\theta {(1-e^{-\beta y_1-\frac{\gamma }{2}y^2_1})}^{{\alpha }_1}{(1-e^{-\beta y_2-\frac{\gamma }{2}y^2_2})}^{{\alpha }_2+{\alpha }_3}\right\}-1}{e^{\theta }-1}
& \ \ {\rm if}\ \ y_1>y_2.
\end{array}
\right.\]
and its pdf is
\begin{eqnarray*}
f_1(y_1,y_2)&=&\theta f_{{\rm G}}(y_1;{\alpha }_1+{\alpha }_3)f_{{\rm G}}(y_2;{\alpha }_2)\exp\{\theta F_{{\rm G}}(y_1;{\alpha }_1+{\alpha }_3)F_{{\rm G}}(y_2;{\alpha }_2)-1\}\\
&&\times\left[\theta F_{{\rm G}}(y_1;{\alpha }_1+{\alpha }_3)F_{{\rm G}}(y_2;{\alpha }_2)+1\right],\\
f_2(y_1,y_2)&=&\theta f_{{\rm G}}(y_1;{\alpha }_1)f_{{\rm G}}(y_2;{\alpha }_2+{\alpha }_3)\exp\{\theta F_{{\rm G}}(y_1;{\alpha }_1)F_{\rm G}
(y_2;{\alpha }_2+{\alpha }_3)-1\}\\
&&\times\left[\theta F_{{\rm G}}(y_1;{\alpha }_1)F_{{\rm G}}(y_2;{\alpha }_2+{\alpha }_3)+1\right],\\
f_0(y)&=&\frac{\theta {\alpha }_3}{\left({\alpha }_1+{\alpha }_2+{\alpha }_3\right)}f_{{\rm G}}\left(y;{\alpha }_1+{\alpha }_2+{\alpha }_3\right)\exp\{\theta F_{{\rm G}}\left(y;{\alpha }_1+{\alpha }_2+{\alpha }_3\right)-1\}.
\end{eqnarray*}

\subsection{Bivariate GLFR-binomial distribution}

When $a_n=\binom{k}{n}$ and $C\left(\theta \right)={\left(\theta +1\right)}^k-1$ ($\theta >0$), where $k$ ($n\le k$) is the number of replicas, the power series distribution becomes the binomial distribution (truncated at zero). Therefore, the cdf of bivariate GLFR-binomial (BGLFRB) distribution is given by
\[F_{{\boldsymbol Y}}(y_1,y_2)=\left\{
\begin{array}{ll}
\frac{\left\{{\theta }{(1-e^{-\beta y_1-\frac{\gamma }{2}y^2_1})}^{{\alpha }_1+{\alpha }_3}{(1-e^{-\beta y_2-\frac{\gamma }{2}y^2_2})}^{{\alpha }_2}+1\right\}^k-1}{{(\theta +1)}^k-1}
& \ \ {\rm if}\ \  y_1\le y_2 \\
\frac{\left\{{\theta }{(1-e^{-\beta y_1-\frac{\gamma }{2}y^2_1})}^{{\alpha }_1}{(1-e^{-\beta y_2-\frac{\gamma }{2}y^2_2})}^{{\alpha }_2+{\alpha }_3}+1\right\}^k-1}{{(\theta +1)}^k-1}
& \ \ {\rm if}\ \  y_1>y_2,
\end{array}
\right.\]
and its pdf is
\begin{eqnarray*}
f_1(y_1,y_2)&=&\frac{k\theta }{{(\theta +1)}^k-1}f_G(y_1;{\alpha }_1+{\alpha }_3)f_{{\rm G}}(y_2;{\alpha }_2){\left[\theta F_{{\rm G}}
(y_1;{\alpha }_1+{\alpha }_3)F_{{\rm G}}(y_2;{\alpha }_2)+1\right]}^{k-2}\\
&&\times \left[k\theta F_{{\rm G}}(y_1;{\alpha }_1+{\alpha }_3)F_{{\rm G}}(y_2;{\alpha }_2)+1\right],\\
f_2(y_1,y_2)&=&\frac{k\theta }{{(\theta +1)}^k-1}f_{{\rm G}}(y_1;{\alpha }_1)f_{{\rm G}}(y_2;{\alpha }_2+{\alpha }_3)
{\left[\theta F_{{\rm G}}(y_1;{\alpha }_1)F_{{\rm G}}(y_2;{\alpha }_2+{\alpha }_3)+1\right]}^{k-2}\\
&&\times\left[k\theta F_{{\rm G}}(y_1;{\alpha }_1)F_{{\rm G}}(y_2;{\alpha }_2+{\alpha }_3)+1\right],\\
f_0(y)&=&\frac{k\theta {\alpha }_3f_{{\rm G}}
(y;{\alpha }_1+{\alpha }_2+{\alpha }_3)}{[{(\theta +1)}^k-1]({\alpha }_1+{\alpha }_2+{\alpha }_3)}{(\theta F_{{\rm G}}(y;{\alpha }_1+{\alpha }_2+{\alpha }_3)+1)}^{k-1}.
\end{eqnarray*}
\subsection{Bivariate GLFR-logarithmic distribution}

When $a_n=\frac{1}{n}$ and $C(\theta)=-{\log\left(1-\theta \right)}$ ($0<\theta <1$), the power series distribution becomes the logarithmic distribution (truncated at zero). Therefore, the cdf of bivariate GLFR-logarithmic (BGLFRL) distribution is given by
\[F_{{\boldsymbol Y}}\left(y_1,y_2\right)=\left\{
\begin{array}{ll}
\frac{{\log  \left(1-\theta {(1-e^{-\beta y_1-\frac{\gamma }{2}y^2_1})}^{{\alpha }_1+{\alpha }_3}{(1-e^{-\beta y_2-\frac{\gamma }{2}y^2_2})}^{{\alpha }_2}\right)}}{{\log  (1-\theta) }}
& \ \ {\rm if}\ \ y_1\le y_2 \\
\frac{{\log  \left(1-\theta {(1-e^{-\beta y_1-\frac{\gamma }{2}y^2_1})}^{{\alpha }_1}{(1-e^{-\beta y_2-\frac{\gamma }{2}y^2_2})}^{{\alpha }_2+{\alpha }_3}\right)}}{{\log  (1-\theta ) }}
& \ \ {\rm if}\ \ y_1>y_2, \end{array}
\right.\]
and its pdf is
\begin{eqnarray*}
f_1(y_1,y_2)&=&\frac{-\theta f_{{\rm G}}(y_1;{\alpha }_1+{\alpha }_3)f_{{\rm G}}(y_2;{\alpha }_2)}{{\log  (1-\theta )}{\left(1-\theta F_{{\rm G}}(y_1;{\alpha }_1+{\alpha }_3)F_{{\rm G}}(y_2;{\alpha }_2)\right)}^2},\\
f_2(y_1,y_2)&=&\frac{-\theta f_{{\rm G}}(y_1;{\alpha }_1)f_{{\rm G}}(y_2;{\alpha }_2+{\alpha }_3)}{{\log  (1-\theta )}{\left(1-\theta F_{{\rm G}}(y_1;{\alpha }_1)F_{{\rm G}}(y_2;{\alpha }_2+{\alpha }_3)\right)}^2},\\
f_0(y)&=&\frac{-\theta {\alpha }_3f_{{\rm G}}(y;{\alpha }_1+{\alpha }_2+{\alpha }_3)}{{\log  (1-\theta ) }({\alpha }_1+{\alpha }_2+{\alpha }_3)\left(1-\theta F_{{\rm G}}(y;{\alpha }_1+{\alpha }_2+{\alpha }_3)\right)}.
\end{eqnarray*}

\subsection{ Bivariate GLFR - negative binomial distribution}

When $a_n=\binom{n-1}{k-1}$ and $C(\theta)=(\frac{\theta }{1-\theta })^k$ ($0<\theta <1$), the power series distribution becomes the negative binomial distribution (truncated at zero). Therefore, the cdf of bivariate GLFR-negative binomial (BGLFRNB) distribution is given by
\[F_{{\boldsymbol Y}}(y_1,y_2)=\left\{
\begin{array}{ll}
\frac{{(1-\theta)}^k{(1-e^{-\beta y_1-\frac{\gamma }{2}y^2_1})}^{k{\alpha }_1+k{\alpha }_3}{(1-e^{-\beta y_2-\frac{\gamma }{2}y^2_2})}^{k{\alpha }_2}}{{\left(1-\theta {(1-e^{-\beta y_1-\frac{\gamma }{2}y^2_1})}^{{\alpha }_1+{\alpha }_3}{(1-e^{-\beta y_2-\frac{\gamma }{2}y^2_2})}^{{\alpha }_2}\right)}^k}
& \ \  {\rm if}\ \ y_1\le y_2 \\
\frac{{(1-\theta )}^k{(1-e^{-\beta y_1-\frac{\gamma }{2}y^2_1})}^{k{\alpha }_1}{(1-e^{-\beta y_2-\frac{\gamma }{2}y^2_2})}^{k{\alpha }_2+k{\alpha }_3}}{{\left(1-\theta {(1-e^{-\beta y_1-\frac{\gamma }{2}y^2_1})}^{\alpha_1}{(1-e^{-\beta y_2-\frac{\gamma }{2}y^2_2})}^{{\alpha }_2+{\alpha }_3}\right)}^k}
& \ \  {\rm if}\ \  y_1>y_2,
\end{array}
\right.\]
and its pdf is
\begin{eqnarray*}
f_1(y_1,y_2)&=&\frac{k{(1-\theta )}^kf_{{\rm G}}(y_1;{\alpha }_1+{\alpha }_3)f_{{\rm G}}(y_2;{\alpha }_2)F^{k-1}_G(y_1;{\alpha }_1+{\alpha }_3)F^{k-1}_{{\rm G}}(y_2;{\alpha }_2)}{{(1-\theta F_{{\rm G}}(y_1;{\alpha }_1+{\alpha }_3)F_{{\rm G}}(y_2;{\alpha }_2))}^{k+2}}\\
&&\times \left[k+\theta F_{{\rm G}}(y_1;{\alpha }_1+{\alpha }_3)F_{{\rm G}}(y_2;{\alpha }_2)\right],\\
f_2(y_1,y_2)&=&\frac{k{(1-\theta )}^kf_{{\rm G}}(y_1;{\alpha }_1)f_{{\rm G}}(y_2;{\alpha }_2+{\alpha }_3)F^{k-1}_{{\rm G}}(y_1;{\alpha }_1)F^{k-1}_{{\rm G}}(y_2;{\alpha }_2+{\alpha }_3)}{{(1-\theta F_{{\rm G}}(y_1;{\alpha }_1)F_{{\rm G}}(y_2;{\alpha }_2+{\alpha }_3))}^{k+2}} \\
&&\times\left[k+\theta F_{{\rm G}}(y_1;{\alpha }_1)F_{{\rm G}}(y_2;{\alpha }_2+{\alpha }_3)\right],\\
f_0(y)&=&\frac{k{\alpha }_3{(1-\theta )}^kf_{{\rm G}}(y;{\alpha }_1+{\alpha }_2+{\alpha }_3)F^{k-1}_{{\rm G}}(y;{\alpha }_1+{\alpha }_2+{\alpha }_3)}{({\alpha }_1+{\alpha }_2+{\alpha }_3){(1-\theta F_{{\rm G}}(y;{\alpha }_1+{\alpha }_2+{\alpha }_3))}^{k+1}}.
\end{eqnarray*}

\section{Estimation}
\label{sec.est}

In this section, we consider the estimation of the unknown parameters of the BGLFRPS distributions. Let
$\left(y_{11},y_{12}\right),\dots ,\left(y_{m1},y_{m2}\right)$ be an observed sample with size $m$ from BGLFRPS distributions with parameters
 ${\boldsymbol \Theta }=\left({\alpha }_1,{\alpha }_2,{\alpha }_3,\beta ,\gamma,\theta \right)'$. Also, consider
\[I_0=\left\{i:y_{1i}=y_{2i}=y_i\right\},\ \  \ \ \ \ I_1=\left\{i:y_{1i}<y_{2i}\right\},\ \ \ \ \ \ I_2=\left\{i:y_{1i}>y_{2i}\right\},\]
and
\[m_0=\left|I_0\right|,\ \ \ \ \ \ m_1=\left|I_1\right|,\ \ \ \ \ \ m_2=\left|I_2\right|,\ \ \ \ \ \ m=m_0+m_1+m_2\]
Therefore, the log-likelihood function can be written as
\begin{equation}\label{eq.lik}
\ell ({\boldsymbol \Theta })=\sum_{i\in I_0}{{\log  (f_0(y_i))}}+\sum_{i\in I_1}{{\log
(f_1(y_{1i},y_{2i})) }}+\sum_{i\in I_2}{{\log  (f_{{\rm 2}}(y_{1i},y_{2i})) }},
\end{equation}
where $f_0$, $f_1$, and $f_2$ are given in
\eqref{eq.f1}, \eqref{eq.f2} and \eqref{eq.f0}, respectively.

We can obtain the maximum likelihood estimations (MLE's) of the parameters by maximizing $\ell \left({\boldsymbol \Theta }\right)$ in \eqref{eq.lik} with respect to the unknown parameters. This is clearly a six-dimensional problem. However, no explicit expressions are available for the MLE's. We need to solve six non-linear equations simultaneously, which may not be very simple. The maximization can be performed using a command like the nlminb routine in the R software
\citep{rdev-14}.
 But, it is related to initial guesses. Therefore, we present an expectation-maximization (EM) algorithm to find the MLE's of parameters.

For given $n$, consider independent random variables $\{Z_i|N=n\}$, $i=1,2,3$ have the GLFR distribution with parameters
$n{\alpha }_i\beta $ and $\gamma$. It is well-known that
\[\left\{Y_1|N=n\right\}=\left\{\max  \left(Z_1,Z_3\right)|N=n\right\},\ \ \ \ \ \
  \left\{Y_2|N=n\right\}=\left\{\max  \left(Z_2,Z_3\right)|N=n.\right\}\]

Assumed that for the bivariate random vector $\left(Y_1,Y_2\right)$, there is an associated random vectors
\[{\Lambda }_1=\left\{ \begin{array}{ll}
0 & Y_1=Z_1 \\
1 & Y_1=Z_2 \end{array}
\right.\ \ \ \ \ \ \ \ {\rm and}\ \ \ \ \ \ \ {\Lambda }_2=\left\{ \begin{array}{cc}
0 & Y_2=Z_1 \\
1 & Y_2=Z_3. \end{array}
\right.\ \]
Note that if $Y_1=Y_2$, then $\Lambda_1={\Lambda}_2=0$. But if $Y_1<Y_2$ or $Y_1>Y_2$, then
$({\Lambda }_1,{\Lambda }_2)$ is missing. If $(Y_1,Y_2)\in I_1$
then the possible values of $({\Lambda }_1,{\Lambda }_2)$ are $(1,0)$ or $(1,1)$,
and If $(Y_1,Y_2)\in I_2$ then the possible values of $({\Lambda }_1,{\Lambda }_2)$
are $(0,1)$ or $(1,1)$ with non-zero probabilities.

We form the conditional `pseudo' log-likelihood function, conditioning on $N$, and then replace $N$ by $E(N|Y_1,Y_2)$. In the E-step of the EM-algorithm, we treat it as complete observation when they belong to $I_0$. If the observation belong to $I_1$, we form the `pseudo' log-likelihood function by fractioning $(y_1,y_2)$ to two partially complete `pseudo' observations of the form $(y_1,y_2,u_1\left({\boldsymbol \Theta }\right))$  and $(y_1,y_2,u_2\left({\boldsymbol \Theta }\right))$, where $u_1\left({\boldsymbol \Theta }\right)$ and $u_2\left({\boldsymbol \Theta }\right)$ are the conditional probabilities that $({\Lambda }_1,{\Lambda }_2)$ takes values $\left(1,0\right)$ and $(1,1)$, respectively. Since
\begin{eqnarray*}
P\left(Z_3<Z_1<Z_2|N=n\right)&=&\frac{{\alpha }_1{\alpha }_2}{\left({\alpha }_1+{\alpha }_3\right)\left({\alpha }_1+{\alpha }_2+{\alpha }_3\right)},\\
P\left(Z_1<Z_3<Z_2|N=n\right)&=&\frac{{\alpha }_2{\alpha }_3}{\left({\alpha }_1+{\alpha }_3\right)\left({\alpha }_1+{\alpha }_2+{\alpha }_3\right)},
\end{eqnarray*}
we have
\[u_1\left({\boldsymbol \Theta }\right)=\frac{{\alpha }_1}{{\alpha }_1+{\alpha }_3},\ \ \ \ \ \ \ u_2\left({\boldsymbol \Theta }\right)=\frac{{\alpha }_3}{{\alpha }_1+{\alpha }_3}.\]
Similarly, If the observation belong to $I_2$, we form the `pseudo' log-likelihood function of the from $\left(y_1,y_2,v_1\left({\boldsymbol \Theta }\right)\right)$ and $\left(y_1,y_2,v_2\left({\boldsymbol \Theta }\right)\right)$, where $v_1\left({\boldsymbol \Theta }\right)$ and $v_2\left({\boldsymbol \Theta }\right)$ are the conditional probabilities that $({\Lambda }_1,{\Lambda }_2)$ takes values $\left(0,1\right)$ and $(1,1)$, respectively. Therefore,
\[v_1\left({\boldsymbol \Theta }\right)=\frac{{\alpha }_2}{{\alpha }_2+{\alpha }_3},\ \ \ \ \ \ \ v_2\left({\boldsymbol \Theta }\right)=\frac{{\alpha }_3}{{\alpha }_2+{\alpha }_3}.\]
For brevity, we write  $u_1\left({\boldsymbol \Theta }\right)$, $u_2\left({\boldsymbol \Theta }\right)$, $v_1\left({\boldsymbol \Theta }\right)$, $v_2\left({\boldsymbol \Theta }\right)$ as $u_1$, $u_2$, $v_1$, $v_2$, respectively.

\noindent \textbf{E-step:} Consider $b_i=E(N|y_{1i},y_{2i},{\boldsymbol \Theta })$. The log-likelihood function without the additive constant can be written as follows:
\begin{eqnarray*}
{\ell }_{{\rm pseudo}}({\boldsymbol \Theta })&=&{\log  (\theta )}\sum^m_{i=1}{b_i}-m{\log  (C(\theta )) }+\sum_{i\in I_0}{{\log  (\beta -\gamma y_i)
}}+\sum_{i\in I_1\cup I_2}{{\log  (\beta -\gamma y_{1i})}}\\
&&+\sum_{i\in I_1\cup I_2}{{\log  (\beta -\gamma {y}_{2i}) }}+{(m}_1u_1{{\rm +}m_2{\rm )log} ({\alpha }_1)}+(m_1+m_2v_1){\log  ({\alpha }_2)}\\
&&+(m_0+m_1u_2+m_2v_2){\log  ({\alpha }_3)}-\beta (\sum_{i\in I_0}{y_i}+\sum_{i\in I_1\cup I_2}{(y_{1i}+y_{2i})})\\
&&-\frac{\gamma }{2}(\sum_{i\in I_0}{y^2_i}+\sum_{i\in I_2}{(y^2_{1i}+y^2_{2i})})\\
&&+{\alpha }_1\left(\sum_{i\in I_0}{b_i{\log  (1-e^{-\beta y_i-\frac{\gamma }{2}y^2_i})}}+\sum_{i\in I_1\cup I_2}{b_i{\log
(1-e^{-\beta y_{1i}-\frac{\gamma }{2}y^2_{1i}}) }}\right)\\
&&+{\alpha }_2\left(\sum_{i\in I_0}{b_i{\log  (1-e^{-\beta y_i-\frac{\gamma }{2}y^2_i})}+}\sum_{i\in I_1\cup I_2}{b_i{\log  (1-e^{-\beta y_{2i}-\frac{\gamma }{2}y^2_{2i}}) }}\right)\\
&&+{\alpha }_3\left(\sum_{i\in I_0}{b_i{\log  \left(1-e^{-\beta y_i-\frac{\gamma }{2}y^2_i}\right)}}+\sum_{i\in I_1}{b_i{\log  \left(1-e^{-\beta y_{1i}-\frac{\gamma }{2}y^2_{1i}}\right) }}\right.\\
&&\left.+\sum_{i\in I_2}{b_i{\log  \left(1-e^{-\beta y_{2i}-\frac{\gamma }{2}y^2_{2i}}\right) }}\right)
-\sum_{i\in I_1}{{\log  \left(1-e^{-\beta y_{1i}-\frac{\gamma }{2}y^2_{1i}}\right)}}\\
&&-\sum_{i\in I_1\cup I_2}{{\log  \left(1-e^{-\beta y_{1i}-\frac{\gamma }{2}y^2_{1i}}\right)}}-\sum_{i\in I_1\cup I_2}{b_i{\log  \left(1-e^{-\beta y_{2i}-\frac{\gamma }{2}y^2_{2i}}\right) }}.
\end{eqnarray*}

\noindent\textbf{M-step: }At this step,  ${\ell }_{{\rm pseudo}}\left({\boldsymbol \Theta }\right)$ is maximized with respect to the parameters. For fixed $\beta $ and $\gamma $, the maximization with respect to ${\alpha }_1,{\alpha }_2$ and ${\alpha }_3$ occurs at
\begin{eqnarray}
&&{\hat{\alpha }}_1\left(\beta ,\gamma \right)=\frac{m_1u_1{\rm +}m_2}{\sum_{i\in I_0}{b_iQ\left(y_i;\beta ,\gamma \right)}+\sum_{i\in I_1\cup I_2}{b_iQ\left(y_{1i};\beta ,\gamma \right)}},\label{eq.a1}\\
&&{\hat{\alpha }}_2\left(\beta ,\gamma \right)=\frac{m_1+m_2v_1}{\sum_{i\in I_0}{b_iQ\left(y_i;\beta ,\gamma \right)+}\sum_{i\in I_1\cup I_2}{b_iQ\left(y_{2i};\beta ,\gamma \right)}},\label{eq.a2}\\
&&{\hat{\alpha }}_3\left(\beta ,\gamma \right)=\frac{m_0+m_1u_2+m_2v_2}{\sum_{i\in I_0}{b_iQ\left(y_i;\beta ,\gamma \right)}+\sum_{i\in I_1}{b_iQ\left(y_{1i};\beta ,\gamma \right)}+\sum_{i\in I_2}{b_iQ\left(y_{2i};\beta ,\gamma \right)}}\label{eq.a3}
\end{eqnarray}
where $Q(y;\beta,\gamma)={\log  (1-e^{-\beta y-\frac{\gamma }{2}y^2})}$.
The maximization of ${\ell }_{{\rm pseudo}}\left({\boldsymbol \Theta }\right)$ can be obtained by solving the non-linear equation
\begin{equation}\label{eq.thh}
\frac{\theta C'\left(\theta \right)}{C\left(\theta \right)}=\overline{b},
\end{equation}
with respect to $\theta $, where $\overline{b}=\frac{1}{m}\sum^m_{i=1}{b_i}$.

Finally, the maximization of ${\ell }_{{\rm pseudo}}\left({\boldsymbol \Theta }\right)$ with respect to $\beta $ and $\gamma $, can be obtained by maximizing ${\ell }_{{\rm pseudo}}\left({\hat{\alpha }}_1\left(\beta ,\gamma \right),{\hat{\alpha }}_2\left(\beta ,\gamma \right),{\hat{\alpha }}_3\left(\beta ,\gamma \right),\beta ,\gamma \right)$, the pseudo-profile log-likelihood function of $\beta $ and $\gamma $.

The following steps can be used to compute the MLE's of the parameters via the EM algorithm:

\noindent \textbf{Step 1}: Take some initial value of ${\boldsymbol \Theta }$\textbf{, }say ${{\boldsymbol \Theta }}^{{\rm (0)}}{\rm =}\left({\alpha }^{(0)}_1,{\alpha }^{(0)}_2,{\alpha }^{\left(0\right)}_3,{\beta }^{\left(0\right)},{\gamma }^{(0)},{\theta }^{(0)}\right)'$.

\noindent \textbf{Step 2}: compute
$b_i=E\left(N| y_{1i},y_{2i};{{\boldsymbol \Theta }}^{\left(0\right)}\right)$

\noindent \textbf{Step 3: }Compute $u_1$, $u_2$, $v_1$, and $v_2$.

\noindent \textbf{Step 4}: maximize the pseudo-profile log-likelihood function
$\ell_{\rm pseudo}(\hat{\alpha}_1(\beta,\gamma),\hat{\alpha}_2(\beta,\gamma),\hat{\alpha}_3(\beta,\gamma),$ $\beta,\gamma)$ with respect to $\beta $ and $\gamma $, say ${\hat{\beta }}^{(1)}$ and ${\hat{\gamma }}^{(1)}$, respectively.

\noindent \textbf{Step 5}: Compute ${\hat{\alpha }}^{\left(1\right)}_i={\hat{\alpha }}_i({\hat{\beta }}^{\left(1\right)},{\hat{\gamma }}^{(1)})$, $i=1,2,3$ from \eqref{eq.a1}-\eqref{eq.a3}.

\noindent \textbf{Step 6}: Find $\hat{\theta }$ by solving the equation \eqref{eq.thh}, say ${\hat{\theta }}^{(1)}$.

\noindent \textbf{Step 7}: Replace ${{\boldsymbol \Theta }}^{{\rm (0)}}$ by ${{\boldsymbol \Theta }}^{{\rm (1)}}{\rm =}\left({\alpha }^{\left(1\right)}_1,{\alpha }^{\left(1\right)}_2,{\alpha }^{\left(1\right)}_3,{\beta }^{\left(1\right)},{\gamma }^{(1)},{\theta }^{\left(1\right)}\right)$, go back to step 1 and continue the process until convergence take place.

\section{A real example}
\label{sec.exa}

The data set was first published in ``Washington Post'' and is available in
\cite{cs-we-89}.
It is obtained from the American Football League for the matches played on three consecutive weekends in 1986. Here, $X_1$ represents the `game time' to the first points scored by kicking the ball between goal posts, and represents the `game time' to the first points scored by moving the ball into the end zone. The data are given in Table \ref{tab.da}.

\begin{table}[ht]
\begin{center}
\caption{Scoring times (in minutes) for the matches.} \label{tab.da}

\begin{tabular}{|c|rrrrrrrrrrrrrr|} \hline
$X_1$&  2.05 &  9.05 & 0.85 & 3.43 &  7.78 & 10.57 &  7.05 &  2.58 &  7.23 &  6.85 & 32.45 &  8.53 & 31.13 & 14.58 \\
$X_2$&  3.98 &  9.05 & 0.85 & 3.43 &  7.78 & 14.28 &  7.05 &  2.58 &  9.68 & 34.58 & 42.35 & 14.57 & 49.88 & 20.57 \\ \hline
$X_1$&  5.78 & 13.80 & 7.25 & 4.25 &  1.65 &  6.42 &  4.22 & 15.53 &  2.90 &  7.02 &  6.42 &  8.98 & 10.15 &  8.87 \\
$X_2$& 25.98 & 49.75 & 7.25 & 4.25 &  1.65 & 15.08 &  9.48 & 15.53 &  2.90 &  7.02 &  6.42 &  8.98 & 10.15 &  8.87 \\ \hline
$X_1$& 10.40 &  2.98 & 3.88 & 0.75 & 11.63 &  1.38 & 10.53 & 12.13 & 14.58 & 11.82 &  5.52 & 19.65 & 17.83 & 10.85 \\
$X_2$& 10.25 &  2.98 & 6.43 & 0.75 & 17.37 &  1.38 & 10.53 & 12.13 & 14.58 & 11.82 & 11.27 &  10.7 & 17.83 & 38.07 \\ \hline
\end{tabular}
\end{center}
\end{table}

We divided all the data by 100. Then, six special cases of BGLFRPS distributions are considered:  BGLFR, BGLFRG, BGLFRP, BGLFRB (with $k=10)$, BGLFRNB (with $k=2$) and BGLFRL. Using the proposed EM algorithm, these models are fitted to the bivariate data set, and the MLE's and their corresponding log-likelihood values are calculated. The results are given in Table \ref{tab.res}.
For each fitted model, the Akaike Information Criterion (AIC), the corrected Akaike information criterion (AICC) and the Bayesian information criterion (BIC) are calculated. We also obtain the Kolmogorov-Smirnov (K-S) distances with the corresponding p-values (in brackets) between the fitted distribution and the empirical cdf for three random variables $Y_1$, $Y_2$  and $\max(Y_1,Y_2)$. Finally, we make use the likelihood ratio test (LRT) for testing the BGE against other models. The statistics and the corresponding p-values are given in Table \ref{tab.res}.

\begin{table}[ht]
\begin{center}
\caption{The MLE's, log-likelihood, AIC, AICC, BIC, K-S, and LRT statistics for six sub-models of BGLFRPS distributions.} \label{tab.res}

\begin{tabular}{|c|c|c|c|c|c|c|} \hline
 & \multicolumn{6}{|c|}{Distribution} \\ \hline
Statistic & BGLFR & BGLFRG & BGLFRP & BGLFRB & BGLFRNB & BGLFRL \\ \hline
${\hat{\alpha }}_1$ & 0.0921 & 0.0605 & 0.0578 & 0.0597 & 0.01955 & 0.0675 \\
${\hat{\alpha }}_2$ & 0.5722 & 0.4197 & 0.3896 & 0.3988 & 0.1325 & 0.4720 \\
${\hat{\alpha }}_3$ & 1.1519 & 0.7471 & 0.7172 & 0.7409 & 0.2421 & 0.8332 \\
$\hat{\beta }$ & 9.6187 & 12.0961 & 11.4616 & 11.2802 & 11.6386 & 12.2489 \\
$\hat{\gamma }$ & $2\times {10}^{-4}$ & $2\times {10}^{-4}$ & $2\times {10}^{-4}$ & $2\times {10}^{-4}$ & $2\times {10}^{-4}$ & $2\times {10}^{-4}$ \\
$\hat{\theta }$ & --- & 0.6128 & 1.9930 & 0.2326 & 0.7186 & 0.8053 \\ \hline

${\log  (\ell )\ }$ & 36.6700 & 38.3625 & 38.2328 & 38.1661 & 38.1721 & 38.3582 \\
AIC & -63.3400 & -64.7250 & -64.4657 & -64.3323 & -64.3443 & -64.7164 \\
AICC & -61.6734 & -62.3250 & -62.0657 & -61.9323 & -61.9443 & -62.3164 \\
BIC & -54.6517 & -54.2990 & -54.0396 & -53.9063 & -53.9183 & -54.2904 \\ \hline
K-S ($Y_1$) & 0.1808 & 0.1880 & 0.1887 & 0.1884 & 0.1890 & 0.1867 \\
(p-value) & (0.1282) & (0.1028) & (0.1005) & (0.1016) & (0.0995) & (0.1071) \\ \hline
K-S ($Y_2$) & 0.1411 & 0.1469 & 0.1507 & 0.1506 & 0.1507 & 0.1422 \\
(p-value) & (0.3408) & (0.2953) & (0.2679) & (0.2688) & (0.2681) & (0.3321) \\ \hline
K-S (${\max  (Y_1,Y_2)\ }$) & 0.1350 & 0.1378 & 0.1428 & 0.1429 & 0.1425 & 0.1325 \\
(p-value) & (0.3929) & (0.3685) & (0.3271) & (0.3262) & (0.3292) & (0.4165) \\ \hline
LRT & --- & 150.0651 & 149.8058 & 149.6724 & 149.6844 & 150.0565 \\
(p-value) & --- & 0.0000 & 0.0000 & 0.0000 & 0.0000 & 0.0000 \\ \hline
\end{tabular}
\end{center}
\end{table}
\newpage

\bibliographystyle{apa}

\end{document}